\documentclass[conference]{IEEEtran}

\usepackage{amssymb}
\usepackage{amsmath}
\usepackage{graphicx}
\usepackage{cite}
\usepackage{subfigure}
\usepackage{graphicx,epstopdf}
\usepackage{epsfig}	
\usepackage{cite,graphicx,amsmath,amssymb,amsthm}
\usepackage{comment}
\usepackage{lipsum}

\newtheorem{theorem}{\textbf{Theorem}}

\newtheorem{lemma}{\textbf{Lemma}}

\newtheorem{corollary}{\textbf{Corollary}}

\makeatletter
\def\ScaleIfNeeded{%
\ifdim\Gin@nat@width>\linewidth \linewidth \else \Gin@nat@width
\fi } \makeatother

\begin{document}
%

\title{\Huge{Safeguarding Massive MIMO Aided HetNets Using Physical Layer Security}}

\author{
\IEEEauthorblockN{ Yansha Deng\IEEEauthorrefmark{1}, Lifeng Wang\IEEEauthorrefmark{2},  Kai-Kit Wong\IEEEauthorrefmark{2}, Arumugam Nallanathan\IEEEauthorrefmark{1},  Maged Elkashlan\IEEEauthorrefmark{3}, \\ and Sangarapillai Lambotharan \IEEEauthorrefmark{4}  }


\IEEEauthorblockA{
\IEEEauthorrefmark{1}Department of Informatics, King's College London, London, UK\\
\IEEEauthorrefmark{2}Department of Electronic and
Electrical Engineering, University College London, London, UK\\
\IEEEauthorrefmark{3}School of Electronic Engineering and Computer
Science, Queen Mary University of London, London, UK\\
\IEEEauthorrefmark{4} Department of Electronic and Electrical Engineering, Loughborough University,
Leicestershire, UK
 }

 }

\maketitle

\begin{abstract}
This paper exploits the potential of physical layer security in massive multiple-input multiple-output (MIMO) aided  two-tier heterogeneous networks (HetNets). We focus on the downlink secure transmission in the presence of multiple eavesdroppers.  We first address the impact of massive MIMO on the maximum receive power based user association. We then derive the tractable upper bound  expressions for the secrecy outage probability of a  HetNets user. We show that the implementation of massive MIMO   significantly improves the secrecy performance, which indicates that physical layer security could be a promising solution for safeguarding massive MIMO HetNets. Furthermore, we show that the secrecy outage probability of  HetNets user first degrades and then improves with increasing the  density of PBSs.
\end{abstract}

\section{Introduction}
Security and privacy in 5G networks is of paramount importance~\cite{5G_vision, Lifeng_commag}. Physical layer security has recently attracted much attention as a potential security solution at the physical layer~\cite{Wyner}. Such security technique exploits propagation randomness to establish secret and avoids using ciphering keys. The FP7 Europe research project PHYLAWS~\cite{PHYLAWS} focuses on the realistic implantation of physical layer security in the existing and future wireless networks. 

Research efforts on the physical layer security have been made by considering different aspects, such as antenna selection~\cite{L_TWC_2015}, cooperative jamming~\cite{GanZheng2011}, and artificial noise~\cite{Yansha_ICC_2015}, etc.  In \cite{Zhu2014}, matched filter precoding and artificial noise generation was designed to secure downlink transmission in a multicell massive multiple-input multiple-output (MIMO) system in the presence of an eavesdropper. In \cite{shaoshi_yan},  physical layer security has been investigated in a two-tier downlink HetNets, where the cooperative femtocells help macrocell achieve the optimal secrecy transmit beamforming.


In 5G, massive multiple-input multiple-output (MIMO) and heterogeneous networks (HetNets) are two key enablers. Physical layer security in massive MIMO enabled HetNets has not been conducted yet and is in its infancy. We believe that massive MIMO enabled HetNets is a new highly rewarding candidate for physical layer security due to the following factors:
\begin{itemize}
  \item \textbf{Base station densities}. In HetNets, different tiers have different base station densities, and small cells are deployed in a large scale to improve the spectrum efficiency. As such, the distance between the user and its serving base station is shorter, which in turn decreases the risk of information leakage.

  \item \textbf{Large antenna arrays}. Base station with large antenna array provides large array gain for its legitimate user. As such, the transmit power level can be cut, and the received signal power at the eavesdropper is correspondingly reduced, due to the fact that the eavesdropper cannot obtain the array gain.

  \item \textbf{Time division duplex}. Massive MIMO is recommended to be applied in time division duplex (TDD) system, to save the pilot resources. In the TDD mode, base station estimates the uplink channel via uplink pilot signals from user, and obtains the downlink channel state information (CSI) based on the channel reciprocity, which means that there is no channel training in the downlink. As such, eavesdropper cannot easily estimate the eavesdropper's channel during the downlink transmission.

\end{itemize}

Motivated by the above, this paper considers physical layer security in the downlink $K$-tier HetNets with massive MIMO, which to the best of our knowledge, has not been studied yet. Each macrocell base station (MBS) is equipped with large antenna arrays and uses linear zero-forcing beamforming (ZFBF) to communicate with dozens of single antenna users over the same time and frequency band. Each picocell base station (PBS) equipped with a single antenna serves one single antenna user for each transmission. We adopt a stochastic geometry approach
to model the different tiers, where the locations of MBSs, PBSs and eavesdroppers are modelled following independent homogeneous Poisson point processes. We first address the impact of massive MIMO on the maximum receive power based user association.  We then derive the upper bound for the secrecy outage probability of a HetNets  user, to show the benefits of  massive MIMO. Our results confirms that using massive MIMO  can significantly enhance the secrecy outage probability of the macrocell user. Furthermore,  the secrecy outage probability of the HetNets  user first increases and then decreases with  increasing the density of  PBSs.
\section{System Model}
In the TDD two-tier HetNets consisting of  macrocells
and  picocells, downlink transmission is considered in the presence of multiple eavesdroppers. Without loss of generality, we assume that the first tier represents the class of MBSs.
The MBSs are located following a homogeneous Poisson point process (HPPP) $\Phi_\mathrm{M}$ with density $\lambda_\mathrm{M}$,
while the PBSs  are located following an independent HPPP $\Phi_\mathrm{P}$ with
density $\lambda_\mathrm{P}$. The eavesdroppers are located following an independent HPPP $\Phi_\mathrm{E}$ with density $\lambda_\mathrm{E}$.

Massive MIMO is adopted in the macrocells~\cite{Jungnickel_IEEE_Commag}, where each $N$-antenna MBS simultaneously communicates with $S$ users $\left(N \gg S \geq 1\right)$, while each PBS and user are single-antenna
nodes. Each MBS uses ZFBF  to transmit $S$ data streams with equal power assignment, such that users that act as
potential malicious eavesdropper can only receive its information signals.
We consider the perfect downlink CSI  and the universal frequency reuse that
all the tiers share the same bandwidth. All the channels undergo independent and identically distributed (i.i.d.) quasi-static Rayleigh fading.

\subsection{User Association}
We consider  user association based on the maximum received power, where a user is associated with the BS that
provides the maximum average received power. The average received power at a user that is connected
with the MBS $\ell$ ($\ell \in \Phi_\mathrm{M}$) is expressed as
\begin{align}\label{Macro_Receive_Power}
{P_{r,\mathrm{M}}} = G_a \frac{P_\mathrm{M}}{S}L\left(\left|X_{\ell,\mathrm{M}}\right|\right),
\end{align}
where $G_a$ is the array gain, $P_\mathrm{M}$ is the MBS's transmit power,
$L\left(\left|X_{\ell,\mathrm{M}}\right|\right)=\beta{ {{\left|X_{\ell,\mathrm{M}}\right|}}^{ - {\alpha_\mathrm{M}}}}$ is
the path loss function, $\beta$ is the frequency dependent constant value, $\left|X_{\ell,\mathrm{M}}\right|$ is the distance,
and $\alpha_{1}$ is the path loss exponent. The array gain $G_a$ of ZFBF transmission
is  ${{N - S + 1}}$~\cite{Hosseini2014_Massive}.

In the picocell, the long-term average received power at a user that is connected with the
PBS $j$ ($j \in \Phi_\mathrm{P}$)  is expressed as
\begin{align}\label{Small_Receive_Power}
{P_{r,\mathrm{P}}} ={P_\mathrm{P}}L\left(\left|X_{j,\mathrm{P}}\right|\right),
\end{align}
where ${P_\mathrm{P}}$ is the PBS's transmit power and $L\left(\left|X_{j,P}\right|\right)=
\beta{\left( {{\left|X_{j,\mathrm{P}}\right|}} \right)^{ - {\alpha_2}}}$ with distance $\left|X_{j,\mathrm{P}}\right|$
and path loss exponent $\alpha_2$.

\subsection{Channel Model}
All the channels undergo the independent and identically distributed (i.i.d.) quasi-static Rayleigh fading. We assume that a typical user is located at the origin $o$. The receive signal-to-interference-plus-noise ratio (SINR) of a typical user at a random distance $\left|X_{o,{\mathrm{M}}}\right|$ from its associated MBS is given by
\begin{align}\label{SINR_Macro}
\mathrm{SINR}_\mathrm{M}= \frac{{\frac{{P_{\mathrm{M}}}}{S} {h_{o,{\mathrm{M}}}} L\left( {\left|{X_{o,{\mathrm{M}}}}\right|} \right)}}{{I_1+ {\delta ^2}}},
\end{align}
where $I_1=I_{\mathrm{M},1}+I_{\mathrm{S},1}$,
$I_{\mathrm{S},1}= {\sum\nolimits_{j \in {\Phi_\mathrm{P}}} {{P_\mathrm{P}}{h_{j,\mathrm{P}}}L\left( {{\left|X_{j,\mathrm{P}}\right|}} \right)} }$,
$I_{\mathrm{M},1} =\sum\nolimits_{\ell  \in {\Phi_\mathrm{M}}\backslash B_{o,\mathrm{M}}} {\frac{P_\mathrm{M}}{S}
{h_{\ell,\mathrm{M}}}L\left( {\left|X_{\ell,\mathrm{M}}\right|} \right)} $,
${h_{o,{\mathrm{M}}}}\sim \Gamma\left(N-S+1,1\right)$ is the small-scale fading channel power gain between the typical user and its associated MBS~\cite{Hosseini2014_Massive},
 $h_{j,\mathrm{P}}\sim \rm{exp}(1)$ and $\left|X_{j,\mathrm{P}}\right|$ are the small-scale fading interfering channel power gain and
distance between the typical user and BS $j$ in the picocell, respectively,  ${h_{\ell,\mathrm{M}}}\sim \Gamma\left(S,1\right)$ and $\left|X_{\ell,\mathrm{M}}\right|$ are the equivalent small-scale fading interfering channel power gain
 and distance between the typical user and MBS $\ell\in {\Phi _\mathrm{M}}\backslash {B_{o,\mathrm{M}}}$ (except the serving BS ${B_{o,\mathrm{M}}}$),
respectively,
 and $\delta ^2$ is the noise power.

The SINR of a typical user at a random distance $\left|X_{o,\mathrm{P}}\right|$ from its associated PBS $B_{o,\mathrm{P}}$ 
 is given by
\begin{align}\label{SINR_Small}
\mathrm{SINR}_\mathrm{P}= \frac{{{P_\mathrm{P}}{g_{o,\mathrm{P}}}L\left( {\left| {{X_{o,\mathrm{P}}}} \right|} \right)}}{{I_2 + {\delta ^2}}},
\end{align}
where $I_2=I_{\mathrm{M},2}+I_{\mathrm{S},2}$, $I_{\mathrm{M},2} = \sum\nolimits_{\ell  \in {\Phi_\mathrm{M}}}
{\frac{P_\mathrm{M}}{S}{g_{\ell,\mathrm{M}}}L\left( {\left| {{X_{\ell,\mathrm{M}}}} \right|} \right)} $,
$I_{\mathrm{S},2} =  {\sum\nolimits_{j \in {\Phi _\mathrm{P}}\backslash {B_{o,\mathrm{P}}}} {{P_\mathrm{P}}{g_{j,\mathrm{P}}}
L\left( {\left| {{X_{j,\mathrm{P}}}} \right|} \right)} }$, $g_{o,\mathrm{P}}\sim \mathrm{exp}(1)$ is the small-scale fading channel power gain between the
typical user and its serving BS, $g_{\ell,\mathrm{M}}\sim \Gamma\left(S,1\right)$ and
$\left| {{X_{\ell,\mathrm{M}}}} \right|$ are the equivalent small-scale fading interfering channel power gain
and distance between the typical user and MBS $\ell$, respectively,
and  $g_{j,\mathrm{P}}$ and $\left| {{X_{j,\mathrm{P}}}} \right|$ are the small-scale fading interfering channel power gain and distance
between the typical user and BS $j\in {\Phi _\mathrm{P}}\backslash {B_{o,\mathrm{P}}}$, respectively, and $g_{j,\mathrm{P}}\sim \mathrm{exp}(1)$.

We consider the non-colluding and passive eavesdropping that each eavesdropper intercepts the signal independently without any attacks. In this case, we only need to focus on the most malicious eavesdropper that has the largest receive SINR. When the MBS transmits the information messages to its intended user, the receive SINR at the most malicious eavesdropper is given by
\begin{align}\label{MBS_Eve_SINR}
{\mathrm{SINR}_{{e^{\rm{*}}}}^\mathrm{M}} = \mathop {\max }\limits_{e \in {\Phi_\mathrm{E}}} \left\{ {\frac{{{\frac{P_\mathrm{M}}{S}}{h_{o,e}}L\left( {\left| {{X_{o,e}}} \right|} \right)}}{{{I_A} + {I_{\mathrm{M},e}} + {I_{\mathrm{S},e}}+\delta^2}}} \right\},
\end{align}
where $h_{o,e} \sim \exp(1)$ and $\left|X_{o,e}\right|$ are the equivalent small-scale fading channel power gain and distance between the eavesdropper and its targeted BS, respectively, $I_A=\frac{P_\mathrm{M}}{S} h_e L\left( {\left| {{X_{o,e}}} \right|} \right)$ with $h_e \sim \Gamma\left(S-1,1\right)$ is the intra-cell interference in the macro cell, ${I_{\mathrm{M},e}}=\sum\nolimits_{\ell  \in {\Phi_\mathrm{M}\setminus o}}
{\frac{P_\mathrm{M}}{S}{h_{\ell,e}}L\left( {\left| {{X_{\ell,e}}} \right|} \right)}$, $h_{\ell,e}\sim \Gamma\left(S,1\right)$ and
$\left| {{X_{\ell,e}}} \right|$ are the equivalent small-scale fading interfering channel power gain
and distance between the eavesdropper and MBS $\ell$, respectively, ${I_{\mathrm{S},e}}=  {\sum\nolimits_{j \in {\Phi_\mathrm{P}}} {{P_\mathrm{P}}{h_{j,\mathrm{P},e}}L\left( {{\left|X_{j,\mathrm{P},e}\right|}} \right)} }$,  $h_{j,\mathrm{P},e}\sim \rm{exp}(1)$ and $\left|X_{j,\mathrm{P},e}\right|$ are the small-scale fading interfering channel power gain and
distance between the eavesdropper and BS $j$ in the picocell, respectively. Similarly. when the PBS transmits the information messages to its intended user,  the receive SINR at the most malicious eavesdropper is given by
\begin{align}\label{SBS_Eve_SINR}
{\mathrm{SINR}_{{e^{\rm{*}}}}^\mathrm{P}} = \mathop {\max }\limits_{e \in {\Phi_\mathrm{E}}} \left\{ {\frac{{{{P_\mathrm{P}}}{g_{o,e}}L\left( {\left| {{X_{o,e}}} \right|} \right)}}{{{I_{\mathrm{M},\mathrm{P},e}} + {I_{\mathrm{S},\mathrm{P},e}}+\delta^2}}} \right\},
\end{align}
where $g_{o,e} \sim \exp(1)$ and $\left|X_{o,e}\right|$ are the equivalent small-scale fading channel power gain and distance between the eavesdropper and its targeted BS, respectively,  ${I_{\mathrm{M},\mathrm{P},e}}=\sum\nolimits_{\ell  \in {\Phi_\mathrm{M}}}
{\frac{P_\mathrm{M}}{S}{g_{\ell,e}}L\left( {\left| {{X_{\ell,e}}} \right|} \right)}$, $g_{\ell,e}\sim \Gamma\left(S,1\right)$ and
$\left| {{X_{\ell,e}}} \right|$ are the equivalent small-scale fading interfering channel power gain
and distance between the eavesdropper and MBS $\ell$, respectively, ${I_{\mathrm{S},\mathrm{P},e}}=  {\sum\nolimits_{j \in {\Phi_\mathrm{P}}\setminus o} {{P_\mathrm{P}}{g_{j,\mathrm{P},e}}L\left( {{\left|X_{j,\mathrm{P},e}\right|}} \right)} }$, $g_{j,\mathrm{P},e}\sim \rm{exp}(1)$ and $\left|X_{j,\mathrm{P},e}\right|$ are the small-scale fading interfering channel power gain and
distance between the eavesdropper and BS $j$ in the picocell, respectively.

\section{Secrecy Performance}

In an effort to assess  the secrecy outage probability of a HetNets user,  we first characterize the impact of massive MIMO on the cell association probability.

\subsection{User Association Probability}
We first derive the PDF of the distance between a typical user and its serving base station in the following two lemmas.
\begin{lemma}\label{Lemma1}
\emph{The PDF of the distance $\left|{X_{o,\mathrm{M}}}\right|$ between a typical user and its
serving MBS $B_{o,\mathrm{M}}$ is given by}
\begin{align}\label{PDF_MBS_distance}
{f_{\left|{X_{o{\rm{,}}\mathrm{M}}}\right|}}\left( x \right) =& \frac{{2\pi {\lambda _\mathrm{M}}}}
{{{\mathcal{A}_\mathrm{M}}}}x\exp \bigg\{  - \pi {\lambda _\mathrm{M}}{x^2}\nonumber\\
&- \pi
{  {\lambda _\mathrm{P}}{{\left( {\frac{{S{P_\mathrm{P}}}}{{\left( {N - S + 1} \right){P_\mathrm{M}}}}} \right)}^{{2 \mathord{\left/
 {\vphantom {2 {{\alpha _2}}}} \right.
 \kern-\nulldelimiterspace} {{\alpha _2}}}}}{x^{{{2{\alpha _1}} \mathord{\left/
 {\vphantom {{2{\alpha _1}} {{\alpha _2}}}} \right.
 \kern-\nulldelimiterspace} {{\alpha _2}}}}}}   \bigg\}.
\end{align}
\emph{ In \eqref{PDF_MBS_distance}, ${{\cal A}_\mathrm{M}}$ is
the probability that a typical user is associated with the MBS}
\begin{align}\label{A_M_pro}
\mathcal{A}_\mathrm{M}=&2\pi {\lambda _\mathrm{M}}\int_0^\infty  r\exp \bigg\{   - \pi {\lambda _\mathrm{M}}{r^2}\nonumber\\
&- \pi
{  {\lambda _\mathrm{P}}{{\left( {\frac{{S{P_\mathrm{P}}}}{{\left( {N - S + 1} \right){P_\mathrm{M}}}}} \right)}^{{2 \mathord{\left/
 {\vphantom {2 {{\alpha _2}}}} \right.
 \kern-\nulldelimiterspace} {{\alpha _2}}}}}{r^{{{2{\alpha _1}} \mathord{\left/
 {\vphantom {{2{\alpha _1}} {{\alpha _2}}}} \right.
 \kern-\nulldelimiterspace} {{\alpha _2}}}}}}  \bigg\}dr.
\end{align}

\end{lemma}

\begin{lemma}
The PDF of the distance $\left|{X_{o,k}}\right|$ between a typical user and its serving PBS in
 $\mathrm{B}_{o,\mathrm{P}}$ is given by
\begin{align}\label{PDF_SBS_distance}
{f_{\left| {{X_{o,\mathrm{P}}}} \right|}}\left( x \right) = &\frac{{2\pi {\lambda _\mathrm{P}}}}{{{{\cal A}_\mathrm{P}}}}x\exp
\bigg\{ - \pi {\lambda _\mathrm{P}}{{x}^{2}}\nonumber\\
& { - \pi {\lambda _\mathrm{M}}\left( {\frac{{{P_\mathrm{M}}\left( {N - S + 1} \right)}}{{{P_\mathrm{P}}S}}} \right){x^{2{\alpha _2}/{\alpha _1}}}}  \bigg\}.
\end{align}
\emph{Here, ${{\cal A}_\mathrm{P}}$ is
the probability that a typical user is associated with the PBS, which is given by}
\begin{align}\label{A_P_pro}
{{\cal A}_\mathrm{P}} = &2\pi {\lambda _\mathrm{P}}\int_0^\infty  r\exp \bigg\{   - \pi {\lambda _\mathrm{P}}{{r}^{2}}\nonumber\\
&{ - \pi {\lambda _\mathrm{M}}\left( {\frac{{{P_\mathrm{M}}\left( {N - S + 1} \right)}}{{{P_\mathrm{P}}S}}} \right){r^{2{\alpha _2}/{\alpha _1}}}}   \bigg\}dr.
\end{align}
\end{lemma}

Note that Lemma 1 and Lemma 2 can be derived following the  approach in \cite{Han-Shin2012}.

\subsection{Achievable Ergodic Rate}
In this subsection, we derive the achievable ergodic rate of the macrocell user and the picocell user.  
\begin{lemma}
For a typical user at a random distance $\left|X_{o,{\mathrm{M}}}\right|$ from its associated MBS, the lower bound on the achievable ergodic rate of the typical macrocell user is derived as
\begin{align}\label{CDF_SINR_M}
&{R_{\rm{M}}^L} = {\log _2}\left( {1 + \frac{{{P_{\rm{M}}}}}{S}\left( {N - S + 1} \right)\beta {{\left( {\frac{{2\pi {\lambda _{\rm{M}}}}}{{{A_{\rm{M}}}}}\Delta } \right)}^{ - 1}}} \right),
\end{align}
where 
\begin{align}\label{delta_1}
\hspace{-0.3cm}\Delta  =&\int_0^\infty  {\left( {\frac{{2\pi {\lambda _{\rm{M}}}{P_{\rm{M}}}\beta {x^{2 - {\alpha _1}}}}}{{{\alpha _1} - 2}} + \frac{{2\pi {\lambda _{\rm{P}}}{P_{\rm{P}}}\beta {{\left( {D_{\rm{P}}^{\rm{M}}} {\left( x \right)} \right)}^{2 - {\alpha _2}}}}}{{{\alpha _2} - 2}} + {\delta ^2}} \right)} 
\nonumber\\ & 
\exp \left\{ { - \pi {\lambda _{\rm{M}}}{x^2} - \pi {\lambda _{\rm{P}}}{{\left( {D_{\rm{P}}^{\rm{M}}}{\left( x \right)} \right)}^2}} \right\}{x^{{\alpha _1} + 1}}dx
.
\end{align}
In \eqref{delta_1},  $D_\mathrm{P}^\mathrm{M} {\left( x \right)} = {\left( {\frac{{S{P_\mathrm{P}}}}{{\left( {N - S + 1} \right){P_\mathrm{M}}}}} \right)^{{1 \mathord{\left/
 {\vphantom {1 {{\alpha _2}}}} \right.
 \kern-\nulldelimiterspace} {{\alpha _2}}}}}{x^{{{{\alpha _1}} \mathord{\left/
 {\vphantom {{{\alpha _1}} {{\alpha _2}}}} \right.
 \kern-\nulldelimiterspace} {{\alpha _2}}}}}
$  is the minimum distance between the interfering picocell BS and the typical marcocell user.
\end{lemma}
\begin{proof}
The achievable ergodic rate of macrocell user is lower bounded by 
\begin{align}\label{bound}
\hspace{-0.3cm} \mathbb{E}\left\{ {{{\log }_2}\left( {1 + {\rm{SIN}}{{\rm{R}}_{\rm{M}}}} \right)} \right\} \ge {R_{\rm{M}}^L} =
{\log _2}\left( {1 + {{\left( {E\left\{ {{\rm{SIN}}{{\rm{R}}_{\rm{M}}}^{ - 1}} \right\}} \right)}^{ - 1}}} \right),
\end{align}
where 
\begin{align}\label{Exp}
 \mathbb{E}\left\{ {{\rm{SIN}}{{\rm{R}}_{\rm{M}}}^{ - 1}} \right\} = & {\left( {\frac{{{P_{\rm{M}}}}}{S}\left( {N - S + 1} \right)\beta } \right)^{ - 1}}
\nonumber\\  & \hspace{-1.5cm}
\int_0^\infty  {\left( {E\left\{ {{I_2}} \right\} + {\delta ^2}} \right)} {x^{{\alpha _1}}}{f_{\left| {{X_{o,{\rm{M}}}}} \right|}}\left( x \right)dx.
\end{align}

In \eqref{Exp}, ${f_{\left| {{X_{o,{\rm{M}}}}} \right|}}\left( x \right)$ is given in \eqref{PDF_MBS_distance}.
Using the Campbell's theorem, the expectation of the aggregate interference from the MBSs and the PBSs is derived as
\begin{align}\label{Exp2}
 \mathbb{E}\left\{ {{I_2}} \right\}  & 
& = \frac{{2\pi {\lambda _M}{P_{\rm{M}}}\beta {x^{2 - {\alpha _1}}}}}{{{\alpha _1} - 2}} + \frac{{2\pi {\lambda _{\rm{P}}}{P_{\rm{P}}}\beta {{\left( {D_{\rm{P}}^{\rm{M}}} {\left( x \right)} \right)}^{2 - {\alpha _2}}}}}{{{\alpha _2} - 2}}.
\end{align}

\end{proof}

\begin{lemma}
For a typical user at a random distance $\left|X_{o,{\mathrm{P}}}\right|$ from its associated PBS,  the achievable ergodic rate of the typical picocell user is derived as
\begin{align}\label{Rate_p}
{R_{\rm{P}}} = \mathbb{E}\left\{ {{{\log }_2}\left( {1 + \rm{SIN{R_{\rm{P}}}}} \right)} \right\} = \frac{1}{{\ln 2}}\int_0^\infty  {\frac{{1 - {\mathbb{F}_{{\rm{SIN}}{{\rm{R}}_{\rm{P}}}}}\left( \gamma  \right)}}{{1 + \gamma }}} d\gamma ,
\end{align}
where
\begin{align}\label{CDF_SINR_P}
&{\mathbb{F}_{{\rm{SIN}}{{\rm{R}}_{\rm{P}}}}}\left( \gamma  \right) =  1 - \frac{{2\pi {\lambda _{\rm{P}}}}}{{{\mathcal{A}_{\rm{P}}}}}\int_0^\infty  {\exp \left\{ { - 2\pi {\lambda _\mathrm{M}}{\Phi _3}\left( x \right) - 2\pi {\lambda _{\rm{P}}}} \right.} 
\nonumber\\ & \frac{{{x^2}\gamma }}{{{\alpha _2} - 2}}{}_2{F_1}\left( {1,1 - \frac{{ - 2}}{{{\alpha _2}}},2 - \frac{2}{{{\alpha _2}}}, - \gamma } \right) - \frac{{{x^{{\alpha _2}}}\gamma {\delta ^2}}}{{{P_{\rm{P}}}\beta }}
\nonumber\\ & \left. { - \pi {\lambda_{\rm{P}}}{x^2} - \pi {\lambda _\mathrm{M}}{{\left( {\frac{{{P_\mathrm{M}}\left( {N - S + 1} \right){x^{{\alpha _2}}}}}{{{P_{\rm{P}}}S}}} \right)}^{2/{\alpha _1}}}} \right\}xdx.
\end{align}
In \eqref{CDF_SINR_P}, we have
\begin{align}\label{Phi_3}
{\Phi _3}\left( x \right) 
 = &{}_2{F_1}\left[ {1 - {2 \mathord{\left/
 {\vphantom {2 {{\alpha _1}}}} \right.
 \kern-\nulldelimiterspace} {{\alpha _1}}},S,2 - {2 \mathord{\left/
 {\vphantom {2 {{\alpha _1}}}} \right.
 \kern-\nulldelimiterspace} {{\alpha _1}}}, - \frac{{\gamma {{\rm{P}}_{\rm{M}}}{x^{{\alpha _2}}}}}{{S{P_{\rm{P}}}{{\left( {D_{\rm{M}}^{\rm{P}}}\left( x \right) \right)}^{{\alpha _1}}}}}} \right]
 \nonumber\\&\frac{{\gamma {P_{\rm{M}}}{x^{{\alpha _2}}}{{\left( {D_{\rm{M}}^P}\left( x \right)\right)}^{2 - {\alpha _1}}}}}{{S{P_{\rm{P}}}\left( {{\alpha _1} - 2} \right)}} 
  + \sum\limits_{k = 2}^S {{S}\choose{k}} \frac{1}{{{\alpha _1}}}{\left( { - \frac{{\gamma {P_{\rm{M}}}{x^{{\alpha _2}}}}}{{S{P_{\rm{P}}}}}} \right)^{{2 \mathord{\left/
 {\vphantom {2 {{\alpha _1}}}} \right.
 \kern-\nulldelimiterspace} {{\alpha _1}}}}}
 \nonumber\\  & B\left( { - \frac{{\gamma {P_{\rm{M}}}{x^{{\alpha _2}}}}}{{S{P_{\rm{P}}}{{\left( {D_{\rm{M}}^{\rm{P}}} \left( x \right)\right)}^{{\alpha _1}}}}};k - {2 \mathord{\left/
 {\vphantom {2 {{\alpha _1}}}} \right.
 \kern-\nulldelimiterspace} {{\alpha _1}}},1 - S} \right),
\end{align}
where  $D_{\rm{M}}^{\rm{P}}\left( x \right) = {\left( {\frac{{\left( {N - S + 1} \right){P_{\rm{M}}}}}{{S{P_{\rm{P}}}}}} \right)^{{1 \mathord{\left/
 {\vphantom {1 {{\alpha _1}}}} \right.
 \kern-\nulldelimiterspace} {{\alpha _1}}}}}{x^{{{{\alpha _2}} \mathord{\left/
 {\vphantom {{{\alpha _2}} {{\alpha _1}}}} \right.
 \kern-\nulldelimiterspace} {{\alpha _1}}}}} 
$ is the minimum distance between the interfering macrocell BS and the typical picocell user,
 $B\left( { \cdot ; \cdot , \cdot } \right)$is the incomplete
beta function \cite[8.391]{gradshteyn}, and ${}_2{F_1}\left[ { \cdot , \cdot , \cdot } \right]$ is the Gauss
hypergeometric function \cite[9.142]{gradshteyn}.
\end{lemma}
\begin{proof}
The CDF of  $\mathrm{SINR}_\mathrm{P}$ is expressed as 
\begin{align}\label{rawCDF_P}
&{\mathbb{F}_{{\rm{SIN}}{{\rm{R}}_{\rm{P}}}}}\left( \gamma  \right) = \int_0^\infty  {\Pr \left[ {{g_{o,P}} \le \frac{{\gamma \left( {{I_2} + {\delta ^2}} \right)}}{{{P_{\rm{P}}}\beta {x^{ - {\alpha _2}}}}}} \right]} {f_{\left| {{X_{o,P}}} \right|}}\left( x \right)dx
\nonumber\\&  = 1 - \int_0^\infty  {\exp \left\{ { - \frac{{\gamma {\delta ^2}{x^{{\alpha _2}}}}}{{{P_{\rm{P}}}\beta }}} \right\}{\mathcal{L}_{{I_2}}}\left( {\frac{{\gamma {x^{{\alpha _2}}}}}{{{P_{\rm{P}}P}\beta }}} \right)} {f_{\left| {{X_{o,P}}} \right|}}\left( x \right)dx
\end{align}
To solve the laplace transform of the aggregate interference from macrocell BSs and picocell BS, we first utilize ${\mathcal{L}_{{I_2}}}\left( {\frac{{\gamma {x^{{\alpha _2}}}}}{{{P_{\rm{P}}}\beta }}} \right) = {\mathcal{L}_{{I_{{\rm{M}},2}}}}\left( {\frac{{\gamma {x^{{\alpha _2}}}}}{{{P_{\rm{P}}}\beta }}} \right){\mathcal{L}_{{I_{{\rm{P}},2}}}}\left( {\frac{{\gamma {x^{{\alpha _2}}}}}{{{P_{\rm{P}}}\beta }}} \right). 
$  The laplace transform of ${I_{{\rm{M}},2}}$ is given by
\begin{align}\label{LIM2}
{\mathcal{L}_{{I_{{\rm{M}},2}}}}\left( s \right) \hspace{-1.3cm}&
\nonumber\\& = {\mathbb{E}_{{I_{{\rm{M}},2}}}}\left\{ {\prod\limits_{\ell  \in {\Phi _{\rm{M}}}} {{\mathbb{E}_g}\left\{ {\exp \left( { - s\frac{{{P_{\rm{M}}}{g_{\ell ,{\rm{M}}}}\beta {{\left| {{X_{\ell ,{\rm{M}}}}} \right|}^{ - {\alpha _1}}}}}{S}} \right)} \right\}} } \right\}
\nonumber\\&
\mathop  = \limits^{\left( a \right)}\exp \left\{ { - 2\pi {\lambda _{\rm{M}}}\int_{D_{\rm{M}}^{\rm{P}}\left( x \right)}^\infty  {\left( {1 - {{\left( {1 + s\frac{{{P_{\rm{M}}}\beta {y^{ - {\alpha _1}}}}}{S}} \right)}^{ - S}}} \right)ydy} } \right\}, 
\end{align}
where  $(a)$ follows from probability generating functional (PGFL) of PPP  \cite{stoyanstochastic} and the Cartesian to polar coordinates transformation.
The Laplace transform of ${I_{{\rm{P}},2}}$ is given by
\begin{align}\label{LIP2}
{\mathcal{L}_{{I_{{\rm{P}},2}}}}\left( s \right) \hspace{-1.2cm}&
\nonumber\\& = {\mathbb{E}_{{I_{{\rm{P}},2}}}}\left\{ {\prod\limits_{j \in {\Phi _{\rm{P}}}\backslash {B_{o,{\rm{P}}}}} {{E_g}\left\{ {\exp \left( { - s{P_{\rm{P}}}{g_{j,{\rm{P}}}}\beta {{\left| {{X_{j,{\rm{P}}}}} \right|}^{ - {\alpha _2}}}} \right)} \right\}} } \right\}
\nonumber\\&
  = \exp \left( { - 2\pi {\lambda _{\rm{P}}}\int_x^\infty  {\left( {1 - {{\left( {1 + s{P_{\rm{P}}}\beta {r^{ - {\alpha _2}}}} \right)}^{ - 1}}} \right)rdr} } \right).
\end{align}
Substituting \eqref{LIM2}  and \eqref{LIP2} into \eqref{rawCDF_P}, we finally derive \eqref{CDF_SINR_P}.
\end{proof}

\subsection{Secrecy Outage Probability}
Secrecy outage probability is the principle  performance metric in the passive eavesdropping scenario.  The secrecy outage is declared when the instantaneous secrecy rate  is less than the targeted secrecy rate $R_s$   \cite{junzhu2014secure}.

\begin{theorem}
For a typical user associated with the MBS, the upper bound on the secrecy outage probability of this typical user is given by 
\begin{align} 
P_{out}^{\rm{M}}\left( {{R_s}} \right) = &\Pr \left\{ {{R_{\rm{M}}} - {{\log }_2}\left( {1 + {\rm{SINR}}_{{e^{\rm{*}}}}^{\rm{M}}} \right) \le {R_s}} \right\}
\nonumber\\
 = & 1 - {\mathbb{F}_{{\rm{SINR}}_{{e^{\rm{*}}}}^{\rm{M}}}}\left( {{2^{\left( {R{_{\rm{M}}} - {R_s}} \right)}} - 1} \right), \label{secrecy_out_M}
\end{align} 
where   $R_{\rm{M}}$ is the lower bound of the  ergodic rate of the macrocell user  in \eqref{bound}, $R_s$ is the targeted secrecy rate, and the CDF of the receive SINR at the most malicious eavesdropper is derived as
\begin{align}\label{CDF_SINR_M_E}
&{\mathbb{F}_{{\rm{SINR}}_{{e^{\rm{*}}}}^{\rm{M}}}}\left( \gamma  \right) =  
\exp \left\{ { - 2\pi {\lambda _{\rm{E}}}\int_0^\infty  {\exp \left\{ { - \gamma S{\delta ^2}{x^{{\alpha _1}}}{{\left( {\beta {P_{\rm{M}}}} \right)}^{ - 1}}} \right.} } \right.
\nonumber \\&  - 2\pi {\lambda _{\rm{M}}}\sum\limits_{k = 1}^S {{{S}\choose{k}}} \frac{{{{\left( {\gamma {x^{{\alpha _1}}}} \right)}^{{2 \mathord{\left/
 {\vphantom {2 {{\alpha _1}}}} \right.
 \kern-\nulldelimiterspace} {{\alpha _1}}}}}\Gamma \left( {k - {2 \mathord{\left/
 {\vphantom {2 {{\alpha _1}}}} \right.
 \kern-\nulldelimiterspace} {{\alpha _1}}}} \right)\Gamma \left( { - k + {2 \mathord{\left/
 {\vphantom {2 {{\alpha _1}}}} \right.
 \kern-\nulldelimiterspace} {{\alpha _1}}} + S} \right)}}{{{\alpha _1}\Gamma \left( S \right)}}
 \nonumber \\&
 \left. { - \frac{{2{\pi ^2}{\lambda _{\rm{P}}}}}{{{\alpha _2}}}{{\left( {\gamma S{x^{{\alpha _1}}}{P_{\rm{P}}}{{\left( {{P_{\rm{M}}}} \right)}^{ - 1}}} \right)}^{{2 \mathord{\left/
 {\vphantom {2 {{\alpha _2}}}} \right.
 \kern-\nulldelimiterspace} {{\alpha _2}}}}}Csc\left[ {{{2\pi } \mathord{\left/
 {\vphantom {{2\pi } {{\alpha _2}}}} \right.
 \kern-\nulldelimiterspace} {{\alpha _2}}}} \right]} \right\}
 \nonumber \\& \left. {{{\left( {\gamma  + 1} \right)}^{ - \left( {S - 1} \right)}}xdx} \right\}
.
\end{align}  


\end{theorem}

\begin{theorem}
For a typical user  associated with the PBS,  the secrecy outage probability of this typical user is derived as 
\begin{align}
P_{out}^{\rm{P}}\left( {{R_s}} \right) =& \Pr \left\{ {{R_{\rm{P}}} - {{\log }_2}\left( {1 + {\rm{SINR}}_{{e^{\rm{*}}}}^{\rm{P}}} \right) \le {R_s}} \right\}
\nonumber\\
 = & 1 - {\mathbb{F}_{{\rm{SINR}}_{{e^{\rm{*}}}}^{\rm{P}}}}\left( {{2^{\left( {R{_{\rm{P}}} - {R_s}} \right)}} - 1} \right) \label{secrecy_out_P},
\end{align} 
where 
 $R_{\rm{P}}$ is the achievable ergodic rate of the  picocell user in  \eqref{Rate_p}, and
the CDF of the receive SINR at the most malicious eavesdropper is given by
\begin{align}\label{CDF_SINR_P_E}
&{\mathbb{F}_{{\rm{SINR}}_{{e^{\rm{*}}}}^{\rm{P}}}}\left( \gamma  \right) =  \exp \left\{ { - 2\pi {\lambda _{\rm{E}}}\int_0^\infty  {\exp \left\{ { - \gamma {\delta ^2}{x^{{\alpha _2}}}{{\left( {\beta {P_{\rm{P}}}} \right)}^{ - 1}}} \right.} } \right.\
\nonumber \\&  - 2\pi {\lambda _M}\sum\limits_{k = 1}^S {{S}\choose{k}}{\left( {\frac{{\gamma {P_{\rm{M}}}}}{{{P_{\rm{P}}}S}}} \right)^{{2 \mathord{\left/
 {\vphantom {2 {{\alpha _1}}}} \right.
 \kern-\nulldelimiterspace} {{\alpha _1}}}}}\frac{{\Gamma \left( { - k + {2 \mathord{\left/
 {\vphantom {2 {{\alpha _1}}}} \right.
 \kern-\nulldelimiterspace} {{\alpha _1}}} + S} \right)\Gamma \left( {k - {2 \mathord{\left/
 {\vphantom {2 {{\alpha _1}}}} \right.
 \kern-\nulldelimiterspace} {{\alpha _1}}}} \right)}}{{{\alpha _1}\Gamma \left( S \right)}}
 \nonumber \\& \left. {\left. {{x^{{{2{\alpha _2}} \mathord{\left/
 {\vphantom {{2{\alpha _2}} {{\alpha _1}}}} \right.
 \kern-\nulldelimiterspace} {{\alpha _1}}}}} - \frac{{2{\pi ^2}{\lambda _{\rm{P}}}}}{{{\alpha _2}}}{\gamma ^{{2 \mathord{\left/
 {\vphantom {2 {{\alpha _2}}}} \right.
 \kern-\nulldelimiterspace} {{\alpha _2}}}}}Csc\left[ {{{2\pi } \mathord{\left/
 {\vphantom {{2\pi } {{\alpha _2}}}} \right.
 \kern-\nulldelimiterspace} {{\alpha _2}}}} \right]{x^2}} \right\}xdx} \right\}
 .
\end{align}

\end{theorem}
The results in  Theorem 1 and Theorem 2 can be derived following the similar approach in the proof for  Lemma 4.

The  secrecy outage probability of the HetNets user is given by 
\begin{align}
{P_{out}}\left( {{R_s}} \right) = P_{out}^\mathrm{M}\left( {{R_s}} \right){\mathcal{A}_\mathrm{M}} + P_{out}^\mathrm{P}\left( {{R_s}} \right){\mathcal{A}_\mathrm{P}}, 
\label{definition}
\end{align}
where 
 ${\mathcal{A}_\mathrm{M}}$ and ${\mathcal{A}_\mathrm{P}}$ are the user cell association probability in the macrocell and the picocell, which are derived in \eqref{A_M_pro}  and \eqref{A_P_pro}.

\section{Numerical Examples}

In this section, we evaluate the achievable ergodic rate and the secrecy outage probability of the considered massive MIMO HetNets based on the analytical results derived in Section III and Monte Carlo simulation. We consider a downlink HetNets  in a circular region with radius 100m. In all simulations, we assume that the network  operates at the carrier frequencey 1GHz,  the bandwidth is 10MHz, the transmit power of the MBS is  ${P_{\mathrm{M}}} = 46$ dBm,  the  transmit power of the PBS is ${P_{\mathrm{P}}}= 37$  dBm, the path loss exponent of macrocell is $\alpha_1 = 3.5$, the path loss exponent of picocell is $\alpha_2 = 4$, the density of MBSs  is $\lambda_{\mathrm{M}} = 10^{-3}$, the density of eavesdroppers  is $\lambda_{\mathrm{E}} = 10^{-1}$,   the users simultaneously served by each MBS is $S=10$,  and the thermal noise is $\sigma^2 = -90$ dBm. Both figures show that the analytical plots have a good match  with the simulation plots.

Fig.~\ref{fig:1} plots the  ergodic rate of the marcocell user and the picocell user versus the number of antennas at each MBS $N$ using \eqref{CDF_SINR_M} and \eqref{Rate_p}. It is shown that the achievable ergodic rate improves with increasing $N$, due to the large array gain. This indicates that massive MIMO MBS carries more data traffic.  Interestingly, increasing the density of PBSs, the achievable ergodic rates of  the marcocell user and the picocell user degrade. This is due to the dominant impact of increased intercell interference brought by the PBSs.

\begin{figure}[t!]
    \begin{center}
        \includegraphics[width=3.0 in,height=2.5in]{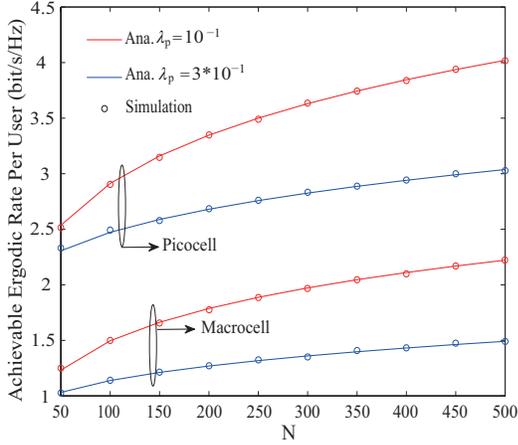}
        \caption{Achievable ergodic rate versus the number of antennas at each MBS}
        \label{fig:1}
    \end{center}
\end{figure}


Fig.~\ref{fig:2}  plots the achievable secrecy outage probability  versus the  density of the PBSs. We set the number of antennas at each MBS as $N = 200$, and define the targeted secrecy rate at the marcocell user as $ R_s = \rho { R_\mathrm{M}}$, and  the targeted secrecy rate at the picocell user as $ R_s = \rho {R_\mathrm{P}}$.  It is assumed that $\rho =0.5$.
We  see that the secrecy outage probability of the picocell user decreases with increasing $\lambda_{\rm{P}}$, due to the fact that more interference results in lowering $\mathrm{SINR}_{\rm{P}}$ in \eqref{SINR_Small}. 


More importantly, the secrecy outage probability of macrocell user first degrades then improves with increasing $\lambda_{\rm{P}}$. The reason is that: 1) 
 Increasing $\lambda_{\rm{P}}$ 
increases  the interference from PBSs, thus greatly degrades 
$\mathrm{SINR}_{\rm{M}}$ in \eqref{SINR_Macro}; 2) For very large density of PBSs, the interference from the interfering PBSs  dominates  ${\mathrm{SINR}_{{e^{\rm{*}}}}^\mathrm{M}}$  in \eqref{MBS_Eve_SINR}.  Increasing  $\lambda_{\rm{P}}$   greatly decreases  the distance between the interfering PBSs and the typical eavesdropper, and thus largely  decreases ${\mathrm{SINR}_{{e^{\rm{*}}}}^\mathrm{M}}$. 


\begin{figure}[t!]
    \begin{center}
        \includegraphics[width=3.0 in,height=2.5in]{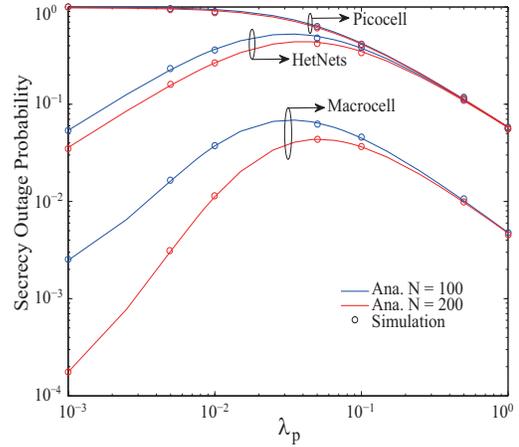}
        \caption{Secrecy outage probability versus the  density of PBSs}
        \label{fig:2}
    \end{center}
\end{figure}

\section{Conclusion}

In this paper, we took into account the physical layer security  for the downlink massive MIMO HetNets with linear zero-forcing beamforming (ZFBF), where the non-colluding  malicious eavesdroppers intercept the downlink  user's  transmission. 
Our work demonstrated the importance of BS deployment density and massive MIMO design on safeguarding the secure downlink transmission in HetNets.

\section{Acknowledgement}
This work was supported by the UK Engineering and Physical Sciences Research Council (EPSRC) with Grant No. EP/M016145/1.

 
\bibliographystyle{IEEEtran}
\bibliography{WCSP_Invited_Paper}

\end{document}